\documentclass[10pt]{article}

\usepackage{a4wide}
\usepackage{hyperref}
\usepackage{authblk}
\usepackage{underscore}
\usepackage[T1]{fontenc}
\usepackage{amsmath,amsthm,amssymb,color}
\usepackage{algorithm}
\usepackage{algpseudocode}
\usepackage{xspace}
\usepackage{physics}
\usepackage{multirow}
\usepackage{tikz,xcolor,amsmath,tkz-euclide}
\usepackage{colortbl}
\usepackage{hhline}
\usepackage[normalem]{ulem}
\usepackage{cleveref}
\usepackage{multicol}
\usepackage{graphicx}
\usepackage{caption}
\usepackage{subcaption}
\usepackage{numprint}
\usepackage{makecell}
\usepackage[showframe=false]{geometry}
\usepackage[backend=bibtex,firstinits,style=numeric,url=false,maxnames=5]{biblatex}

\AtEveryBibitem{
  \clearfield{day}
  \clearfield{month}
  \clearfield{endday}
  \clearfield{endmonth}
}
\DeclareFieldFormat[misc,inproceedings]{date}{\mkbibparens{#1}}  
\DeclareFieldFormat{titlecase}{\MakeSentenceCase*{#1}}

\providecommand{\keywords}[1]{\small\textbf{\textit{Keywords---}} #1}

\newcommand{\PDFstr}[2]{\texorpdfstring{#1}{#2}}
\newcommand{\qbin}[2]{\begin{bmatrix}{#1}\\ {#2}\end{bmatrix}_q}

\newcommand{\np}[1]{\numprint{#1}}

\algnewcommand\algorithmicforeach{\textbf{for each}}
\algdef{S}[FOR]{ForEach}[1]{\algorithmicforeach\ #1\ \algorithmicdo}
\algnewcommand{\IIf}[1]{\State\algorithmicif\ #1\ \algorithmicthen}
\algnewcommand{\EndIIf}{\unskip\ \algorithmicend\ \algorithmicif}
\algnewcommand{\IFor}[2]{\State\algorithmicforeach\ #1\ \algorithmicdo #2\ \algorithmicend\ \algorithmicfor}

\theoremstyle{plain}
\newtheorem{theorem}{Theorem}
\newtheorem{corollary}[theorem]{Corollary}
\newtheorem{lemma}[theorem]{Lemma}
\newtheorem{proposition}[theorem]{Proposition}

\theoremstyle{definition}
\newtheorem{definition}[theorem]{Definition}

\theoremstyle{remark}

\renewcommand*\npstyleenglish{%
\npthousandsep{\,}%
\npdecimalsign{{\cdot}}%
\npproductsign{\times}%
\npunitseparator{\,}%
\npdegreeseparator{}%
\npcelsiusseparator{\nprt@unitsep}%
\nppercentseparator{\nprt@unitsep}%
}

\begin{document}

\title{New and improved bounds on the contextuality degree of multi-qubit configurations\footnote{Published by Cambridge University Press in Mathematical Structures in Computer Science, \url{https://doi.org/10.1017/S0960129524000057}}}

\author[1]{Axel Muller}
\author[2]{Metod Saniga}
\author[1]{Alain Giorgetti\footnote{Corresponding author, \texttt{alain.giorgetti@femto-st.fr}}}
\author[3]{Henri de Boutray}
\author[4,5]{\\Frédéric Holweck}

\affil[1]{Université de Franche-Comté, CNRS, institut FEMTO-ST, F-25000 Besançon, France}
\affil[2]{Astronomical Institute of the Slovak Academy of Sciences, 059 60 Tatranska Lomnica, Slovakia} %
\affil[3]{ColibrITD, France}
\affil[4]{ICB, UMR 6303, CNRS, University of Technology of Belfort-Montbéliard, UTBM, 90010 Belfort, France}
\affil[5]{Department of Mathematics and Statistics, Auburn University, Auburn, AL, USA}

\date{}

\maketitle

\begin{abstract}
We present algorithms and a C code to reveal quantum contextuality and evaluate
the contextuality degree (a way to quantify contextuality) for a variety of
point-line geometries located in binary symplectic polar spaces of small rank.
With this code we were not only able to recover, in a more efficient way, all
the results of a recent paper by de Boutray et al [(2022). Journal of Physics A:
Mathematical and Theoretical 55 475301], but also arrived at a bunch of new
noteworthy results. The paper first describes the algorithms and the C code.
Then it illustrates its power on a number of subspaces of symplectic polar
spaces whose rank ranges from 2 to 7. The most interesting new results
include: (i) non-contextuality of configurations whose contexts are subspaces of
dimension 2 and higher, (ii) non-existence of negative subspaces of dimension
3 and higher, (iii) considerably improved bounds for the contextuality
degree of both elliptic and hyperbolic quadrics for rank 4, as well as for a
particular subgeometry of the three-qubit space whose contexts are the lines of
this space, (iv) proof for the non-contextuality of perpsets and, last but not
least, (v) contextual nature of a distinguished subgeometry of a multi-qubit
doily, called a two-spread, and computation of its contextuality degree.
Finally, in the three-qubit polar space we correct and improve the contextuality
degree of the full configuration and also describe finite geometric
configurations formed by unsatisfiable/invalid constraints for both types of
quadrics as well as for the geometry whose contexts are all 315 lines of the
space. \end{abstract}


\section{Introduction}

Quantum contextuality (see, e.\,g.,~\cite{bcgkl} for a recent comprehensive
review) is an important property of quantum mechanics saying that measurements
of quantum observables cannot be regarded as revealing certain preexisting
values; slightly rephrased, the result of a measurement depends on the context
to which it is associated with. One of the simplest proofs of contextuality of
quantum mechanics is the so-called Mermin--Peres square~\cite[]{mermin,peres},
which demonstrates that already in the Hilbert space of two qubits it is not
possible to consistently assign certain preexisting values to all observables.
Another well-known (observable-based) contextuality proof is furnished by a
Mermin pentagram~\cite[]{mermin} living in the three-qubit Hilbert space. It is
interesting to realize that both the Mermin--Peres square and the Mermin pentagram
are related to distinguished subgeometries of the symplectic polar spaces
$\mathcal{W}(2N-1,2)$ associated with corresponding generalized $N$-qubit Pauli
groups: the former being isomorphic to a geometric hyperplane of
$\mathcal{W}(3,2)$~\cite[]{spph} and the latter living in a doily-related pentad of
Fano planes of $\mathcal{W}(5,2)$~\cite[]{ls,saniga}. These two and several other
(e.\,g.,~\cite{psh,sl}) geometrical observations have recently prompted
us~\cite[]{DHGMS22} to have a more systematic look at such geometrically backed
quantum contextual configurations (i.e., sets of contexts) of
$\mathcal{W}(2N-1,2)$ for some small values of $N$. This paper can be regarded
as an organic continuation of~\cite{DHGMS22} with substantially improved
algorithms and a more efficient code to address contextuality issues
not only for all the types of geometric hyperplanes, but also for a great number
of subspaces of varying dimensions found in multi-qubit symplectic spaces of
rank up to 7.

The paper is organized as follows. \Cref{backgroundSec} summarizes the most
essential properties of binary symplectic polar spaces and their relation with
the generalized $N$-qubit Pauli groups and then recalls the definition of the
degree of contextuality for a quantum configuration. Our contributions begin
with \Cref{methodSec}, that presents algorithms and programs that we created to
generate quantum configurations, and reveal their contextuality or evaluate
their degree of contextuality. In particular, we describe in detail two new
algorithms, to generate totally isotropic subspaces and to compute the degree of
contextuality using a SAT solver. The main computational outcomes of our work
are gathered in \Cref{resultSec}. \Cref{propertiesSec} then establishes or
discusses some general facts based on our findings. Here, we notably prove that
there are no negative subspaces of dimension three and higher, introduce a new
doily-based contextual configuration -- the so-called two-spread -- and compute
its contextuality degree and, last but not least, we show that the contextuality
degree of the configuration comprising all 315 three-qubit lines is 63, which is
much lower than the previously known value. \Cref{conclusionSec} concludes the
paper with some open questions and an outline of prospective tasks.

\section{Background}
\label{backgroundSec}

In this section, we collect all the necessary concepts and introduce the symbols 
and notation to the extent to make the paper as self-contained as possible so 
that the reader should be able to follow the main line of our reasoning without 
the urgent need to consult relevant references (like~\cite[Section~2]{SdHG21} 
and~\cite[Introduction]{MSGDH}).
For all integers $0 \leq k \leq 2N$, the
$2N$-dimensional vector space $\mathbb{F}_2^{2N}$ over the 2-element field
$\mathbb{F}_2=(\{0,1\},+,\times)$ has vector subspaces of dimension $k$. Among
them, the ones that are totally isotropic, without their $0$, form the
\emph{symplectic space} $\mathcal{W}(2N-1,2)$, whose name is shortened as $W_N$.
A subspace is \emph{totally isotropic} if any two vectors $x$ and $y$ in it are
mutually orthogonal ($\innerproduct{x}{y} = 0$), for the symplectic form
$\innerproduct{.}{.}$ defined by

\begin{equation} 
\innerproduct{x}{y} = x_1y_{N+1} + x_{N+1}y_1 + x_2y_{N+2} + x_{N+2}y_2 + \dots 
              + x_{N} y_{2N} + x_{2N} y_{N}.
\label{symplf}
\end{equation}

\noindent In other words, a (totally isotropic) subspace of $W_N$ of
(projective) dimension $1 \leq k \leq N-1$ is a totally isotropic vector
subspace of $\mathbb{F}_2^{2N}$ of dimension $k+1$ without its zero vector. A
\emph{point} of $W_N$ is the unique element of a (totally isotropic)
subspace of $W_N$ of (projective) dimension 0. A \emph{line} of
$W_N$ is a (totally isotropic) subspace of $W_N$ of
(projective) dimension 1. A \emph{generator} of $W_N$ is a (totally
isotropic) subspace of $W_N$ of (projective) dimension $N-1$.

It is known that for a symplectic polar space $\mathcal{W}(2N-1,q)$ embedded in a
projective space PG$(2N-1,q)$, the number of its $k$-dimensional subspaces is
given by (see, e.\,g.,~\cite[Lemma 2.10]{deboeck})
\begin{equation}
  \qbin{N}{k+1}  \prod_{i=1}^{k+1} (q^{N+1-i} +1),
  \label{sub-sympl}
\end{equation}
where 
\begin{equation}
\qbin{n}{k} = \prod_{i=1}^{k} \frac{q^{n-k+i} -1}{q^i -1} = \frac{(q^n-1) \dots (q^{n-k+1} -1)}{(q^k-1) \dots (q-1)}
\label{gauss}
\end{equation}
is the Gaussian (binomial) coefficient.

Let
\begin{equation*}
X = \left(
\begin{array}{rr}
0 & 1 \\
1 & 0 \\
\end{array}
\right),~~
Y = \left(
\begin{array}{rr}
0 & -\text{i} \\
\text{i} & 0 \\
\end{array}
\right)~~{\rm and}~~
Z = \left(
\begin{array}{rr}
1 & 0 \\
0 & -1 \\
\end{array}
\right)
\label{paulis}
\end{equation*}
be the Pauli matrices, $I$ the identity matrix, ``$\otimes$'' the tensor product
of matrices, and $I^{\otimes N}\equiv I_{(1)}\otimes I_{(2)}\otimes\ldots\otimes
I_{(N)}$. From here on, all tensor products of observables $G_1 \otimes G_2
\otimes \cdots \otimes G_N$ are called \emph{$N$-qubit observables} and denoted
$G_1 G_2 \cdots G_N$, by omitting the symbol $\otimes$ for the tensor product.
Let ``$.$'' denote the matrix product and $M^2$ denote $M.M$. It is easy to check
that $X^2 = Y^2 = Z^2 = I$, $X.Y = \text{i}Z = -Y.X$, $Y.Z = \text{i}X = -Z.Y$,
and $Z.X = \text{i}Y = -X.Z$. The $N$-qubit observables and the multiplicative
factors $\pm 1$ and $\pm \text{i}$, called \emph{phase}, form the
(\emph{generalized}) ($N$\emph{-qubit}) \emph{Pauli group} ${\cal P}^{\otimes N}
= (\{1,-1,\text{i},-\text{i}\}\times\{I,X,Y,Z\}^{\otimes N},.)$.

The $N$-qubit observable $G_1 G_2 \cdots G_N$, with $G_j \in \{I, X, Y, Z \}$
for $j \in \{1, 2, \ldots, N \}$, is bijectively represented by the bitvector
$(g_1, g_2, \ldots, g_{2N})$ such that $G_j \leftrightarrow (g_j, g_{j+N})$ for
$j \in \{1, 2,\dots,N\}$, with the assumption that $I\leftrightarrow(0,0)$,
$X\leftrightarrow (0,1)$, $Y\leftrightarrow(1,1)$ and $Z\leftrightarrow (1,0)$.

The non-zero bitvectors are
the points of $W_N$, and two observables commute if and only if their
encodings are orthogonal. Therefore, we sometimes write that two points of
$W_N$ ``commute''.

The \emph{perpset} of the point $p\in W_N$ is the point-line geometry
whose points are all points $q$ which commute with $p$ ($\innerproduct{p}{q}=0$)
and whose lines are all the lines of $W_N$ which contain $p$ (and thus
contain only points in the perpset).

A \emph{quantum configuration} is a pair $(O,C)$ where $O$ is a finite set of
observables and $C$ is a finite set of subsets of $O$, called \emph{contexts}, 
such that (i) each observable $M \in O$ satisfies $M^2=I^{\otimes N}$ (so, its
eigenvalues are in $\{-1,1\}$); (ii) any two observables $M$ and $M'$ in the same
context commute, that is, $M.M'=M'.M$; (iii) the product of all observables in
each context is either $I^{\otimes N}$ (\emph{positive} context) or $-I^{\otimes N}$
(\emph{negative} context). This sign is encoded by the \emph{context valuation}
$e : C \rightarrow \{-1,1\}$ of $(O,C)$, defined by $e(c) = 1$ if the context
$c$ is positive, and $e(c) = -1$ if it is negative.

The points and lines of $W_2$ form a noticeable quantum configuration
called the \emph{two-qubit doily}. A $N$-qubit doily is a quantum configuration
of $W_N$ isomorphic to the two-qubit doily. A \emph{spread of
generators} is a set of pairwise disjoint generators that partition the set of
points of $W_N$. A \emph{two-spread} is a quantum configuration
obtained by removing from a multi-qubit doily one spread of its lines, while
keeping all its points.

\subsection{Contextuality degree}
\label{sec:degree}

Since the contextuality degree is the central subject of the present paper, this
section recalls its definition and its connection with the phenomenon of quantum
contextuality, both introduced in a former work~\cite[]{DHGMS22}.

Let $(O,C)$ be a quantum configuration with $p = |O|$ observables
$\{M_1,\ldots,M_p\}$ and $l = |C|$ contexts $\{c_1,\ldots,c_l\}$.  Its
\emph{incidence matrix} $A \in \mathbb{F}_2^{l \times p}$ is defined by $A_{i,j}
= 1$ if the $i$-th context $c_i$ contains the $j$-th observable $M_j$.
Otherwise, $A_{i,j} = 0$. Its \emph{valuation vector} $E \in \mathbb{F}_2^{l}$
is defined by $E_i = 0$ if $e(c_i) = 1$ and $E_i = 1$ if $e(c_i) = -1$, where
$e$ is the context valuation of $(O,C)$.

With these notations, the quantum configuration $(O,C)$ is contextual iff the
linear system
\begin{equation}
\label{eq:all-pc}
A x = E
\end{equation}
has no solution in $\mathbb{F}_2^{p}$. To be non-contextual means there is a
solution $x\in \mathbb{F}_2^l$ that satisfies~(\ref{eq:all-pc}). Such a
solution $x$ can be thought of as a set of predefined values for the
observables of $O$ that satisfies all constraints imposed by the
contexts of the configuration. Note that the predefined values defined by $x$ do not
depend on the contexts. When such a solution $x$ exists, one says that there exists a
Non-Contextual Hidden Variables (NCHV) model that reproduces the outcomes of the
configuration predicted by quantum mechanics~\cite[]{bcgkl,DHGMS22}.

In this setting, one can measure how much contextual a given quantum
configuration is. Let us denote by $\text{Im}(A)$ the image of the matrix $A$ as
a linear map $A:\mathbb{F}_2^{p}\to \mathbb{F}_2^l$. Then, if $(O,C)$ is
contextual, necessarily $E \notin \text{Im}(A)$. A natural
measure is the degree of contextuality~\cite[]{DHGMS22}, defined as follows:

\begin{definition}[Contextuality degree]
Let $(O,C)$ be a contextual configuration with the valuation vector $E\in 
\mathbb{F}_2^l$. Let us denote by $d_H$ the Hamming distance on the vector space
$\mathbb{F}_2^l$. Then one defines the degree $d$ of contextuality of $(O,C)$ by:
 \begin{equation}
  d=d_H(E,\text{Im}(A)).
 \end{equation}
\end{definition}

The notion of degree of contextuality measures in some sense how far a given
configuration is from being satisfied by an NCHV model. The Hamming distance
tells us what is the minimal number of constraints on the contexts that one
should change to make the configuration valid by a deterministic function. In
other words, $l-d$ measures the maximum number of constraints of the contextual
configuration that can be satisfied by an NCHV model.

The degree of contextuality has also a concrete application as it can be used to
calculate the upper bound for contextual inequalities. Let us consider a
contextual configuration $(O,C)$ and let us denote by $C^+$ the subset of positive
contexts, $C^-$ the subset of negative ones and by $\langle c\rangle$ the
expectation for an experiment corresponding to the context $c$. The following
inequality was established by Cabello~\cite[]{Cab10}:
 \begin{equation}
  \sum_{c\in C^+} \langle c\rangle-\sum_{c \in C^-} \langle c\rangle \leq b.
  \label{noncontextIneq}
 \end{equation}
Under the assumption of quantum mechanics, the upper bound $b$ is the number of
contexts of $(O,C)$, that is, $b=l$. However this upper bound is lower for
contextual configurations under the hypothesis of an NCHV model. Indeed, as shown
in~\cite{Cab10}, this bound is $b=2s-l$, where $s$ is the maximum number of
constraints of the configuration that can be satisfied by an NCHV model. It
connects the notion of degree of contextuality with the upper bound $b$:
\begin{equation}
b=l-2d.
\label{eq:contextinequality}
\end{equation}

\section{Methodology}
\label{methodSec}

This section details the algorithms and programs that we propose to generate
quantum configurations and evaluate their degree of contextuality. All our
programs are written in C language because it is faster and more accessible than
the Magma language of the previous implementation~\cite[]{DHGMS22}. The algorithm
for generating perpsets and hyperbolic and elliptic quadrics is presented
in~\cite{DHGMS22}, but we now implement it in C language. This algorithm
consists of first generating all the lines and then selecting among them the
ones belonging to a given perpset or quadric.

For the generation of totally isotropic subspaces, including lines and
generators, we propose a more efficient algorithm than in~\cite{DHGMS22},
detailed in~\Cref{subspaceAlgoSec}. The approach we propose to check
contextuality (\Cref{contextualOrNotSATsec}) and compute the contextuality
degree (\Cref{,cdeg2SATsec}) is to use a SAT
solver.

\subsection{Generation of totally isotropic subspaces}
\label{subspaceAlgoSec}

Algorithm~\ref{subspaceAlgo} presents the recursive function
\textsc{TotallyIsotropicSubspaces}$(N,k)$ that generates the configuration
of the points of $W_N$ whose contexts are all the totally isotropic
subspaces of dimension $k$ of $W_N$, for $N \geq 2$ and $1 \leq k \leq
N-1$. Thus, the function call \textsc{TotallyIsotropicSubspaces}$(N,N-1)$ builds
the generators of $W_N$, and the function call
\textsc{TotallyIsotropicSubspaces}$(N,1)$ builds the configuration of
$W_N$ whose contexts are the lines of $W_N$.

The parameters $S = \{\}$ and $l = 0$ with a default value are only useful for
the recursive calls and do not have to be specified when calling the function.
The default values mean that calling \textsc{TotallyIsotropicSubspaces}$(N,k)$
automatically initializes $S$ with the empty set and $l$ with 0. The subsequent
recursive calls will then assign different values to $S$ and $l$.

\begin{algorithm}[ht!]
  \begin{center}
    \begin{algorithmic}[1]
      \Function{TotallyIsotropicSubspaces}{$N,k,S = \{\},l = 0$}
       \IIf{$l = k+1$}
          \Return $\{S\}$
       \EndIIf;
       \State{$T \gets \{\};$}
       \ForEach{$p > \textit{max}(S)$ in $W_N$} \label{forPline}

        \State{$C \gets S \cup \{p\}$;}
        \State{$\textit{valid} \gets \textit{true}$;}
        \ForEach{$q \in S$}
          \IIf{$\innerproduct{p}{q} = 1$ or $p+q < p$}\quad$\textit{valid} \gets \textit{false};${\quad\textbf {break}}\EndIIf;\label{filter:subspace}
          \State{$C \gets C \cup \{ p + q \}$}
        \EndFor;
        \IIf{$\textit{valid}$} 
          {$T \gets T \cup \textsc{TotallyIsotropicSubspaces}(N,k,C,l+1)$}
        \EndIIf \label{ifValidLine}
       \EndFor;\label{end:subspace}
       \State \Return $T$
      \EndFunction
    \end{algorithmic}
  \end{center}
  \caption{Totally isotropic subspaces building algorithm.\label{subspaceAlgo}}
\end{algorithm}

The strict total order $<$ on points used in this algorithm is the lexicographic
order, when points are considered as bitvectors (of even length).
The expression $\textit{max}(S)$ represents the
point $p$ in $S$ such that $\forall q \in S, p \geq q$. For the case of the
empty set, it is assumed that $p > \textit{max}(\{\})$ holds for all $p$ in
$W_N$. Let us denote by $S+p$ the set $\{q + p \mid q \in S\}$ for any
point $p$ and any set of points $S$.

The following definition and facts justify some expected properties of the
algorithm.

The \emph{effective size} $t(p)$ of a point $p$ is the size of the associated
vector when we remove all the 0's from the beginning. For instance,
$t((0,0,0,0,0,1,0,1)) = t((1,0,1)) = 3$.

Let $B$ be a basis of the subspace $S$ (of dimension $l-1$) at the
beginning of the function. When \textit{valid} is true on~\Cref{ifValidLine},
the variable $C$ stores the subspace $S \cup \{p\} \cup (S+p)$, of dimension
$l$. One of its bases is $B \cup \{p\}$. So, the algorithm is correct.

Conversely, let $C$ be any totally isotropic subspace of dimension $l$. Let $t$
be the maximal effective size of the elements of $C$. Let $S$ be the subset of
the elements of $C$ whose effective size is strictly less than $t$. Let $m$ be
the minimal point of $C$ whose effective size is $t$. Then, $C$ will be built
on~\Cref{ifValidLine} by the function call
\textsc{TotallyIsotropicSubspaces}$(N,k,S,l-1)$, when the value of $p$
on~\Cref{forPline} is $m$. This remark justifies that the algorithm is complete
and does not generate duplicates of the same subspace.

\subsection{Computing the contextuality degree with a SAT solver}
\label{cdeg2SATsec}

This section shows how the problem of finding the contextuality degree of a
configuration can be reduced to a 3-SAT problem and thus take advantage of the
numerous optimizations implemented in the current SAT solvers. More precisely,
we limit the effort of reduction to 3-SAT by going through an intermediary
translation into the rich input language of the tool
\texttt{bc2cnf}~\cite[]{JT00}, which creates itself a file in DIMACS format
(suffix \texttt{.cnf}), the default format for SAT solvers.

The linear system $Ax = E$ is a set of equations in $\mathbb{F}_2$ (see the
example for a Mermin--Peres square in~\Cref{cdegeq:lin}, where $+$ is the
``exclusive or'' symbol for exclusive disjunction). Each equation corresponds to
a context of the configuration, with the expected value on the right-hand side.
Since there are only zeros and ones, we rephrase the problem of knowing whether
a configuration is contextual as finding a set of Boolean variables $v_i$ ($1
\leq i \leq 9$ in~\Cref{cdegeq:lin}) such that the exclusive disjunction of the
variables is 0 for each positive context and 1 for each negative context. We can
then use any SAT solver to compute a degree of contextuality, by using a feature
of \texttt{bc2cnf} which allows us to specify bounds for the number of
constraints that have to be satisfied.

\begin{figure}
  \centering
  \begin{subfigure}[b]{0.45\textwidth}
    \begin{align*}
      v_1 + v_2 + v_3 = 0\\
      v_4 + v_5 + v_6 = 0\\
      v_7 + v_8 + v_9 = 0\\
      v_1 + v_4 + v_7 = 0\\
      v_2 + v_5 + v_8 = 0\\
      v_3 + v_6 + v_9 = 1\\
  \end{align*}
  \caption{linear system\label{cdegeq:lin}}
  \end{subfigure}
  \hfill
  \begin{subfigure}[b]{0.45\textwidth}
    \begin{verbatim}
      BC1.1
      ASSIGN[5,6](
      v1 ^ v2 ^ v3 == F,
      v4 ^ v5 ^ v6 == F,
      v7 ^ v8 ^ v9 == F,
      v1 ^ v4 ^ v7 == F,
      v2 ^ v5 ^ v8 == F,
      v3 ^ v6 ^ v9 == T
      );
      \end{verbatim}
      \caption{\texttt{bc2cnf} translation\label{cdegeq:bc}}
  \end{subfigure}
  \caption{(a) Contextuality linear system for a Mermin--Peres square and (b) the 
    corresponding \texttt{bc2cnf} file to decide whether 5 to 6 out of the 6
    constraints can be solved (i.e., the contextuality degree is at most $1$).}
  \label{cdegeq}
\end{figure}

A \texttt{bc2cnf} input (see an example in~\Cref{cdegeq:bc}) mainly consists of
constraints. The exclusive disjunction is represented by \verb+^+, the true and
false values respectively, by \verb+T+ and \verb+F+, and the equality of
expressions by \verb+==+. After the prefix \verb+BC1.1+ specifying the version,
the \verb+ASSIGN+ clause encapsulates the constraints under the form
\verb+[low,high](formula)+ where \verb+low+ and \verb+high+ respectively are the
lower and upper bounds for the number of constraints to satisfy. After some
experimentations, we found that \texttt{kissat\_gb}~\cite[]{CSMMYJ} was the fastest
SAT solver for these kinds of systems.

\begin{algorithm}[!ht]
    \begin{center}
      \begin{algorithmic}[1]
        \Function{ContextualityDegree}{$C$}
        \State{$i \gets |C^+|;$}
         \While{$i < |C|$}
            \State $sol \gets sat\_sol(i+1,|C|,C);$\label{sol:cdegalgo}
            \IIf{$sol = \emptyset$} \Return $|C|-i$ \EndIIf;\label{bc:cdegalgo}
            \State{$i \gets n\_match(sol,C)$}
         \EndWhile;
         \State \Return $0$
        \EndFunction
      \end{algorithmic}
    \end{center}
    \caption{Algorithm for calculating the degree of contextuality of a 
      configuration $C$.\label{cdegalgo}}
\end{algorithm}

The \Cref{cdegalgo} computes the contextuality degree of any configuration. The
variable $C$ is the set of the constraints of a configuration, the integer $i$
is the current maximal known number of satisfiable constraints, and $sol$ stores
a solution for the set of constraints.

The $sat\_sol$ function called on~\Cref{sol:cdegalgo} is such that
$sat\_sol(i,j,C)$ retrieves a solution satisfying at least $i$ and at most $j$
of the constraints in the set $C$. The $n\_match(sol,l)$ returns the number of
constraints of $sol$ that are satisfied by $l$.

Let $|C^+|$ denote the number of positive constraints. We know that we can at
least satisfy this number of constraints, by assigning the value 0 to all
variables in $C$. For every iteration of the loop, we know that at least $i$
constraints can be satisfied, so on~\Cref{sol:cdegalgo} we ask the solver if at
least one more constraint than $i$ can be satisfied. If yes, we retrieve a
solution $sol$ of the SAT solver and then count the number of satisfied
constraints, which can be more than $i+1$. This number is then used as the new
limit $i$ until no solution can be found, which means that $|C|-i$ is the
minimal number of unsatisfiable constraints in the system, that is, the
contextuality degree of the corresponding configuration.

\subsection{Checking contextuality with a SAT solver}
\label{contextualOrNotSATsec}

When computing the contextuality degree is too time-consuming, we simply check
contextuality. Checking contextuality of a quantum configuration is showing that
a linear system $A x = E$ has no solution in $\mathbb{F}_2^{p}$, for a matrix
$A$ of size $l\times p$ with $l\leq p$. The complexity of linear system
resolution (\textit{e.g.}, by Gaussian elimination) is polynomial,
 and it is exponential for SAT solving. So, we also
implemented the contextuality test (i.e., whether a configuration is contextual
or not) of subspaces by linear system resolution, with Gaussian elimination
optimized for sparse matrices, as implemented in the function \texttt{lu\_solve}
of the \texttt{SpaSM} tool~\cite[]{spasm}, to compare its efficiency in practice
with SAT solving with \texttt{kissat\_gb}, when computing times are neither
negligible nor too long for this comparison.

For subspaces of dimension $k=2$ and $N=5$ qubits, the computation time is 3~s
for \texttt{kissat\_gb} and 40~s for \texttt{SpaSM}. For $k=3$, it is 45~s for
\texttt{kissat\_gb} and 150~s for \texttt{SpaSM}. This suggests that SAT solving
is faster than Gaussian elimination to check non-contextual configurations.
However, for contextual configurations, \texttt{SpaSM} appears to be faster. For
$N=7$, \texttt{SpaSM} validates the contextuality of an hyperbolic quadric in
21~s, against 32~s for \texttt{kissat\_gb}, and the contextuality of an elliptic
quadric in 14~s, against 31~s for \texttt{kissat\_gb}.

\section{Computational Results}
\label{resultSec}

The results provided by the C program are summarized in~\Cref{resultsTable},
where $N$ is the number of qubits, \# contexts is the number of contexts (i.e.,
the number of rows of $A$ and $E$), \# neg. contexts is the number of contexts
with negative sign among them, \# obs. is the number of observables (i.e., the
number of columns of $A$), $d$ is the contextuality degree, and $D(N)$ is the
number of $N$-qubits doilies, given in~\cite[Section 3]{MSGDH}.

Results are grouped by family of configurations (first column), and \# config°
is the number of configurations in this family. The numbers of subspaces
computed by the C program for any $N$ and $k$ match those given
by~(\ref{sub-sympl}). The column C/NC indicates whether all the configurations
in this family are Contextual or Not Contextual. Since the generators of
$W_2$ are its lines, their properties given in the first block for
lines are not repeated in the block for generators.

\begin{table}[hbt!]
\begin{scriptsize}
\begin{tabular}{|p{2.2cm}|c|r|r|r|r|c|l|}
\hline
Configuration       & $N$&\# neg. contexts&\# config°&\# contexts  &\# obs. & C/NC & $d$\\
\hline
\hline
\multirow{4}{*}{Lines ($k=1$)}
                    &  2 & 3              & 1     & 15             & 15   & C  & 3\\
                    &  3 & 90             & 1     & 315            & 63   & C  & $\mathbf{63}$ ({\it 90})\\
                    &  4 & \np{1908}      & 1     & \np{5355}      & 255  & C  & $\leq \np{2268}$\\
                    &  5 & \np{35400}     & 1     & \np{86955}     & \np{1023} & C  & $\leq \np{40391}$\\
\hline
\multirow{4}{*}{Subspaces ($k=2$)}
                    &  4 & \np{4752}      & 1     & \np{11475}     & 255  & \textbf{NC} & $\mathbf{0}$\\
                    &  5 & \np{358560}    & 1     & \np{782595}    & \np{1023}& \textbf{NC} & $\mathbf{0}$\\
                    &  6 & \np{24330240}  & 1     & \np{50868675}  & \np{2047}& ? & ?\\
                    &  7 & \np{1602215424}& 1     & \np{3268162755}& \np{4095}& ? & ?\\
\hline
\multirow{2}{*}{Subspaces ($k=3$)}
                    &  5 & 0              & 1     & \np{782595}    & \np{1023}& \textbf{NC} & $\mathbf{0}$\\
                    &  6 & 0              & 1     & \np{213648435} & \np{2047}& \textbf{NC} & $\mathbf{0}$\\
\hline
Subspaces ($k=4$)   &  6 & 0              & 1     & \np{103378275} & \np{2047}& \textbf{NC} & $\mathbf{0}$\\
\hline
\multirow{4}{*}{Generators}
                    &  3 & 54             & 1      & 135            & 63   & NC & 0\\
\multirow{4}{*}{($k=N-1$)}
                    &  4 & 0              & 1      &\np{2295}       & 255  & NC & 0\\
                    &  5 & 0              & 1      &\np{75735}      & \np{1023}& NC & 0\\
                    &  6 & 0              & 1      &\np{4922775}    & \np{2047}& \textbf{NC} & $\mathbf{0}$\\
                    &  7 & 0              & 1      &\np{635037975}  & \np{4095}& \textbf{NC} & $\mathbf{0}$\\
\hline
Doilies             & [2..5] & [3..12]    & $D(N)$ & 15             & 15   & C  & 3\\
\hline
Mermin--Peres       & \multirow{2}{*}{[2..5]} & \multirow{2}{*}{1, 3 or 5}  &\multirow{2}{*}{$10\times D(N)/4^{N-2}$}&\multirow{2}{*}{6}& \multirow{2}{*}{9}    & \multirow{2}{*}{C}  & \multirow{2}{*}{1}\\
squares             &        &            &                       & &      &    & \\
\hline
Two-spreads         & [2..5] & 1, 3, 5, 7 or 9& $6\times D(N)$& 10 & 15   & \textbf{C}  & $\mathbf{1}$\\
\hline
\multirow{6}{*}{Hyperbolics}
                    &  2 & 1 or 3         & 10     & 6              & 9    & C  &  1\\
                    &  3 & 27 or 39       & 36     & 105            & 35   & C  &  21\\
                    &  4 & 532            & 81     & \np{1575}      & 135  & C  & $\mathbf{\leq 500}^1$\\
                    &  4 & 604            & 54     & \np{1575}      & 135  & C  & $\mathbf{\leq 500}^1$\\
                    &  4 & 612            & 1      & \np{1575}      & 135  & C  & $\mathbf{\leq 517}^1$\\
                    &  5 & \np{9420}, \np{9852} or \np{9900} & 528& \np{23715}     & 527  & C  & $\leq \np{10878}^1$ \\
\hline
\multirow{5}{*}{Elliptics}
                    &  2 & 0              & 6      & 0              & 5    & C  &  N/A\\
                    &  3 & 9 or 13        & 28     & 45             & 27   & C  &  9\\
                    &  4 & 360            & 12     & 1071           & 119  & C  & $\mathbf{\leq 351}^1$\\
                    &  4 & 384            & 108    & 1071           & 119  & C  & $\mathbf{\leq 363}^1$\\
                    &  5 & \np{7860}, \np{7876} or \np{8020} & 496 & \np{19635}     & 495  & C  & $\leq \np{9169}^1$ \\
\hline
\multirow{6}{*}{Perpsets}
                    &  2 & 0 or 1         & 15     & 3              & 7    & NC &  0\\
                    &  3 & 0, 4 or 6      & 63     & 15             & 31   & NC &  0\\
                    &  4 & 0, 16, 24 or 28& 255    & 63             & 127  & NC &  0\\
                    &  5 & 0 \ldots 120   & \np{1023}   & 255            & 511  & NC &  0\\
                    &  6 & 0 \ldots 496   & \np{4095}   & \np{1023}           & \np{2047} & \textbf{NC} & \textbf{0}\\
                    &  7 & 0 \ldots \np{2016}  & \np{16383}& \np{4095}        & \np{8191} & \textbf{NC} & \textbf{0}\\
\hline
\multicolumn{8}{l}{$^1$ These numbers are found at least for one of the configurations.}\\
\end{tabular}
\end{scriptsize}
\caption{Known and \textbf{new} results for the contextuality degree of various 
  configurations\label{resultsTable}}
\end{table}

\Cref{resultsTable} significantly extends Table~2 in~\cite{DHGMS22}, in which
the bounds for the contextuality degree $d$ for lines with $N = 3, 4, 5$ qubits
should be, according to~\cite{Cab10}, the numbers \textit{90}, \textit{1908}, and
\textit{35400} of negative contexts, instead of the numbers \textit{135},
\textit{1539}, and \textit{16155} of the corresponding bound $b$, with the
notations of~\Cref{sec:degree}. 

For a quicker orientation, the new results and improved bounds, detailed in the
rest of this section, are shown in the \textbf{bold font style}
in~\Cref{resultsTable}.

\subsection{Results for subspaces}

For lines with 3 qubit observables, the C program found a valuation with 63
invalid constraints, and no valuation with 62 invalid constraints or less, thus
justifying that the contextuality degree is 63. This exact value is lower that
the previously known bound, which was the number 90 of negative lines. For 4 and
5 qubits the contextuality checks succeed, but the time to compute the
contextuality degree is too high when the variable $i$ in \Cref{cdegalgo} is
initialized with the number $|C^+|$ of positive contexts (Line 2). When
initializing it with the total number $|C|$ of contexts, we obtained at best the
bounds \numprint{2268} and \numprint{40391} for the contextuality degree for 4
and 5 qubits. These bounds are not interesting, since they are higher than the
number of negative contexts.

For planes ($k = 2$), the contextuality degree has been computed up to five qubits
and is always 0.
\Cref{resultsTable} also shows the number of negative planes for six and seven qubits
that we were able to generate. However we could not compute their contextuality
degree, because the input files for \texttt{bc2cnf} and \texttt{SpaSM} are too
large.

The totally isotropic subspaces of dimension $k$ of $W_N$ computed for
$(k,N) = (3,4)$, $(3,5)$, $(3,6)$, $(4,5)$, $(4,6)$, $(5,6)$ and $(6,7)$ appear
to be positive. These results suggest that this positivity property might hold
for all $N > k \geq 3$, with the consequence that the configuration whose
contexts are all these subspaces for given $k$ and $N$ is non-contextual. After
further examination, this property indeed holds. We state it
as~Proposition~\ref{kgeq3pos-prop} and prove it in~\Cref{tisGeq3Sec}.

\subsection{Results for perpsets, quadrics, doilies, Mermin--Peres squares, and 
  two-spreads}

For quadrics and perpsets, the C program gives the same contextuality answers as
in~\cite{DHGMS22}. Moreover, it additionally checks that all perpsets with 6 and
7 qubits are non-contextual.

For $N$ qubits, there seems to be $N$ possible numbers of negative lines in a
perpset~\cite[]{DHGMS22}. These numbers seem to be 4 times the number of negative
lines in a perpset of $N-1$ qubits, plus an additional number that seems to
match the OEIS A006516 sequence (\url{https://oeis.org/A006516}).

The case of {\it three-qubit} quadrics is particularly interesting in the sense
that, in addition to reconfirming their exact degree of contextuality found
earlier, we can also explicitly see the geometric patterns formed by
corresponding invalid (or unsatisfiable) constraints. Thus, for each of 28
elliptic quadrics living in $\mathcal{W}(5,2)$, we found that 9 invalid
constraints form a set of 9 pairwise disjoint lines partitioning the set of
points of the quadric, that is, it is a {\it spread of lines}. Similarly, for each
of 36 hyperbolic quadrics living in $\mathcal{W}(5,2)$ the geometry formed by
the 21 unsatisfiable contexts is isomorphic to that depicted in Figure~\ref{plig}. 
This geometry is underpinned by the {\it Heawood graph}~\cite[]{heawood}. This 
remarkable bipartite graph has 14 vertices and 21 edges, being isomorphic to the
point-line incidence (or Levi) graph of the Fano plane; in Figure~\ref{plig} the
7 big empty circles correspond to the seven points of the Fano plane, while the seven
big filled circles correspond to the seven lines, with an edge joining an empty
vertex to a filled vertex iff the corresponding point lies on the corresponding
line. On a hyperbolic quadric, both seven-sets can be labeled by the
points/observables of two distinct disjoint Fano planes lying on the quadric in
such a way that any two observables joined by an edge commute and thus define
the line/context whose the third observable (represented by a smaller
gray-shaded circle) is the product of the two. Note that, as in the case of
elliptic quadrics, the 21 lines entail all the 35 points/observables of the
quadric.

\begin{figure}[pth!]
\centerline{\includegraphics[width=6truecm,clip=]{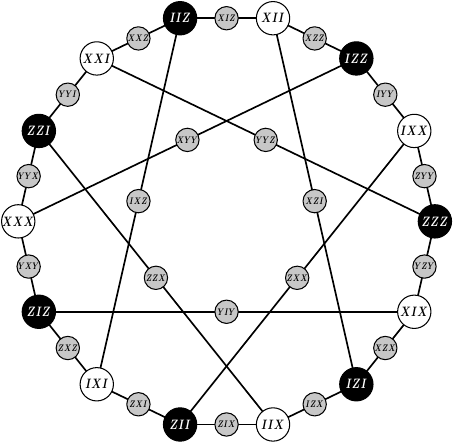}}
\vspace*{.2cm}
\caption{An illustration of the point-line configuration formed by 21 invalid 
  constraints of the particular three-qubit hyperbolic quadric that consists
  solely of symmetric observables.}
\label{plig}
\end{figure}

Concerning Mermin--Peres squares, the table just restates a well-known result
proved in \cite{HS17} that each grid is contextual once embedded in an arbitrary
multi-qubit symplectic polar space of order 2. Similarly, all multi-qubit
doilies have already been shown to be contextual for any $N$, their degree of
contextuality being always 3~\cite[Proposition 1]{MSGDH}. Contextuality of
two-spreads is an important novelty of this work, their degree of contextuality
being proved (see~\Cref{twospreadSec}) to be 1 irrespectively of the  value of the rank
of the ambient symplectic space.

The C program took on average 10 min to compute the contextuality degree 21
of one hyperbolic quadric of 3 qubits, with the default SAT solver
\texttt{minisat}. Almost the entirety of the time is taken checking the
unsatisfiability of the last system. This is the longest computation time,
compared to perpsets and elliptics. For this reason we checked the top three SAT
solvers of 2021
~({\url{https://github.com/simewu/top-SAT-solvers-2021}}) 
and found out that \texttt{kissat\_gb} was the fastest for our problem. Indeed,
it took an average of 10 s to find the contextuality degree of an
hyperbolic of three qubits. Thus, the time to compute the contextuality degree of
all the hyperbolics for three qubits was reduced down to 5 min, and to a few
seconds for elliptics and perpsets.

The results are the same as the ones in~\cite{DHGMS22}. However, as with the
Magma implementation, it does not seem possible to compute in the same way the
contextuality degree of configurations with more than three qubits. From our
computations up to now, it comes that for four qubits the contextuality degree of
elliptics is at most 351 for one of those with 360 negative lines, at most 363
for one of those with 384 negative lines. It is at most 500 for one of the hyperbolics 
with 532 or 604 negative lines, and at most 517 for the single one with 612 negative lines. 
For five qubits, when starting from the number of positive
contexts, the times to compute the contextuality degree of hyperbolics and
elliptics are too high. When starting from the total number of contexts, the
best bounds \numprint{10878} and \numprint{9169} obtained are not interesting,
since they exceed the number of negative contexts  of hyperbolics and elliptics,
respectively.

\section{Contextuality-Related Properties}
\label{propertiesSec}

This section presents contextuality-related properties holding for all numbers
of qubits beyond a certain known threshold.

\subsection{Perpsets}

\begin{proposition}\label{perpsetNCprop}
All perpsets are non-contextual.
\end{proposition}

\begin{proof}
Remember that the perpset of the point $p\in W_N$ is the configuration
of the points $q$ commuting with $p$ whose contexts are the lines of
$W_N$ containing $p$. Consequently, any point $q \neq p$ in the
perpset of $p$ belongs to the unique line $\{p,q,p+q\}$. Thus, the two variables
$v_q$ and $v_{p+q}$ (other than the variable $v_p$ associated with $p$) in the
sign constraint $v_p + v_q + v_{p+q} = e$ of the line $\{p,q,p+q\}$ appear in no
other sign constraint. Whatever the values of $v_p$ and $e$, it is always
possible to find a value for $v_q$ and $v_{p+q}$ to satisfy this constraint $v_p
+ v_q + v_{p+q} = e$. So, the entire system of sign constraints is always
solvable and all perpsets are non-contextual.
\end{proof}

\subsection{Totally isotropic subspaces of dimension \PDFstr{$k \geq 3$}{k ≥ 3}}
\label{tisGeq3Sec}

\begin{proposition}\label{kgeq3pos-prop}
For $N > k \geq 3$, all totally isotropic subspaces of $W_N$ of
dimension $k$ are positive. Consequently, the quantum configuration whose
contexts are the totally isotropic subspaces of $W_N$ of dimension $k$
is non-contextual.
\end{proposition}

\begin{proof}
A totally isotropic subspace of $W_N$ is positive if the phase of the
product of the elements in the $N$-qubit Pauli group that are bijectively
associated with its points is $1$. Thus, after introducing some notations, we
reformulate and formalize the positivity property with the matrix product in this
(non-commutative) group.

Let $M = \{I,X,Y,Z\}$. Let $|\_|$ denote the norm on the Pauli group ${\cal
P}^{\otimes N} = \{1,-1,\text{i},-\text{i}\} \times M^{\otimes N}$, defined as
phase removal by $|p\,a| = a$ for any phase $p \in \{1,-1,\text{i},-\text{i}\}$
and any $a \in M^{\otimes N}$. Let us denote by $|.|$ the binary operation on
$M^{\otimes N}$ such that, for all $a, b \in M^{\otimes N}$, $a~|.|~b = |a.b|$.
(Here, the dot ``$.$'' denotes the matrix product in ${\cal P}^{\otimes N} $.).
For $\iota$ in $\mathbb{F}_2$ (whose elements are the binary digits $0$ and $1$)
and $a$ in $M^{\otimes N}$, let $a^{\iota}$ be the \emph{multiplication by a
scalar} defined by $a^{0} = I^{\otimes N}$ and $a^{1} = a$. With this
multiplication by a scalar and the product $|.|$, it is easy to check that
$M^{\otimes N}$ is a vector space (of dimension $2N$) over $\mathbb{F}_2$.

A totally isotropic subspace of $W_N$ of (projective) dimension $k$ is
bijectively associated with a vector subspace $S$ of (vectorial) dimension $k+1$
of $M^{\otimes N}$, without its neutral element $I^{\otimes N}$. Since
multiplying with this neutral element does not change the value of the product
of all the elements of $S$, the totally isotropic subspace of $W_N$ is
positive if and only if the matrix product (with ``$.$'') of all the observables in
$S$ equals $I^{\otimes N}$.

We illustrate the notations and proof steps with the explicit example
in~\Cref{posProofExample}, for $k=3$ and $N=4$.

\begin{table}[htb!]
\resizebox{\columnwidth}{!}{%
{$\arraycolsep=0.4pt\def\arraystretch{1.7}
\begin{array}{|c|c|c|c|c|c|c|c|c|c|c|}
\hline 
 i
  & i_3 
    i_2
    i_1
    i_0
   & b^i
    & |b^i|_3
     & |b^{2i}|_3.|b^{2i+1}|_3&|b^i|_2
      & |b^{2i}|_2.|b^{2i+1}|_2
       & |b^i|_1
        & |b^{2i}|_1.|b^{2i+1}|_1
         & |b^i|_0
          & \parbox[c][1cm]{3.2cm}{$|b^{2i}|_0.|b^{2i+1}|_0.|b^{2r}.b^{2i}|_0.$\\ $|b^{2r}.b^{2i}.b_0|_0$}
\\
\hline
 0
  & 0000
   & IIII
    & I
     &
      & I
       &
        & I
         &
          & I
\\
\cline{1-3}
 1
  &  0001
   &  b_0 = IIYZ 
    & I
     & \multirow{-2}{*}{$I.I=I$}
      & I
       & \multirow{-2}{*}{$I.I=I$}
        &  Y
         & \multirow{-2}{*}{$I.Y=Y$}
          & Z
\\
\hline
 2
  &  0010
   &  b_1 = IXXX
    & I
     & 
      & X
       &
        & {\scriptsize X}
         & \cellcolor{lightgray}
          & X
           &
\\
\cline{1-3}
 3
  &  0011
   & b_1.b_0 =  IXZY  
    & I
     & \multirow{-2}{*}{$I.I=I$}
      & X
       & \multirow{-2}{*}{$X.X=I$}
        & Z
         & \cellcolor{lightgray}\multirow{-2}{*}{$X.Z =-\text{i}Y$}
          & Y
           &  
\\
\hhline{----------}
 4
  &  0100
   &  b_2 = XZXY
    & X
     &
      & Z
       &
        & X
         & \cellcolor{lightgray}
          & Y
           & \cellcolor{white}
\\
\cline{1-3}
 5
  &  0101
   & b_2.b_0 = -XZZX  
    & X& \multirow{-2}{*}{$X.X=I$}
     & Z& \multirow{-2}{*}{$Z.Z = I$}
      & Z
       & \cellcolor{lightgray}\multirow{-2}{*}{$X.Z =-\text{i}Y$}
        & X
         & \multirow{-4}{*}{$X.Y.Y.X = I$}
\\
\hline
 6
  &  0110
   & b_2.b_1 = XYIZ  
    & X
     &
      & Y
       &
        & I
         &
          & {\scriptsize Z}
\\
\cline{1-3}
 7
  &  0111
   & b_2.b_1.b_0 = XYYI 
    & X
     & \multirow{-2}{*}{$X.X=I$}
      & Y
       & \multirow{-2}{*}{$Y.Y = I$}
        & Y
         & \multirow{-2}{*}{$I.Y = Y$}
          & I
\\
\hhline{----------}
 8
  &  1000
   &  b_3 = ZIXX 
    & Z
     &
      & I
       &
        & X
         & \cellcolor{lightgray}
          & X
\\
\cline{1-3}
9 
&  1001  
  & b_3.b_0 = ZIZY
  & Z
    & \multirow{-2}{*}{$Z.Z=I$} 
    & I 
      & \multirow{-2}{*}{$I.I = I$} 
      & Z 
        & \cellcolor{lightgray}\multirow{-2}{*}{$X.Z =-\text{i}Y$} 
        & Y
\\
\hline
 10 
  &  1010  
   & b_3.b_1 = ZXII
    & Z
     &  
      & X 
       &  
        & I 
         &  
          & I 
           &
\\ 
\cline{1-3}
 11 
  &  1011  
   & b_3.b_1.b_0 = ZXYZ
    & Z
     & \multirow{-2}{*}{$Z.Z=I$} 
      & X 
       & \multirow{-2}{*}{$X.X = I$} 
        & Y 
         & \multirow{-2}{*}{$I.Y = Y$} 
          & Z
           & 
\\ 
\hhline{----------}
 12 
  &  1100  
   & b_3.b_2 = -YZIZ
    & Y
     &  
      & Z 
       &  
        & I 
         &  
          & Z 
           & \cellcolor{white}
\\
\cline{1-3}
 13 
  &  1101  
   & b_3.b_2.b_0 = -YZYI
    & Y
     & \multirow{-2}{*}{$Y.Y=I$} 
      & Z 
       & \multirow{-2}{*}{$Z.Z = I$} 
        & Y 
         & \multirow{-2}{*}{$I.Y = Y$} 
          & I
           & \multirow{-4}{*}{$I.Z.Z.I=I$}
\\ \hline
 14 
  &  1110  
   & b_3.b_2.b_1 = YYXY
    & Y
     &  
      & Y 
       &  
        & X 
         & \cellcolor{lightgray} 
          & Y
\\ 
\cline{1-3} 
 15 
  &  1111  
   & b_3.b_2.b_1.b_0 = -YYZX
    &Y
     & \multirow{-2}{*}{$Y.Y=I$} 
      & Y 
       & \multirow{-2}{*}{$Y.Y = I$} 
        & Z 
         & \cellcolor{lightgray} \multirow{-2}{*}{$X.Z =-\text{i}Y$} 
          & X
\\
\hline
\multicolumn{3}{c|}{} 
 & \multicolumn{2}{c|}{I^8} 
  & \multicolumn{2}{c|}{I^8} 
   & \multicolumn{2}{c|}{Y^4.(-\text{i}Y)^4=I} 
    & \multicolumn{2}{c|}{(I.Z).I.(Z.I).(X.Y).I.(Y.X)=I}
\\
\cline{4-11}
\multicolumn{3}{c|}{} 
 & \multicolumn{2}{c|}{\text{trivial}} 
  & \multicolumn{2}{c|}{\text{trivial}} 
   & \multicolumn{2}{c|}{\text{case~1, with }U=X} 
    & \multicolumn{2}{c|}{\text{case~2, with } r = (0,1,1) \text{ and } q = 1}
\\
\cline{4-11}
\end{array}
$
}}

\caption{Illustrative example for the proofs of~Proposition~\ref{kgeq3pos-prop} and~Lemma~\ref{qubitPosLemma}.\label{posProofExample}}
\end{table}

For $m \geq 0$, the natural number $i = i_{m} 2^{m} + \ldots i_0 2^0$, with
$i_{m}$, \ldots, $i_0 \in \mathbb{F}_2$, is identified with the tuple
$(i_{m},\ldots,i_0)$ of its binary digits.  For $m=3$, these 16
numbers $i \in [0..15]$ and their binary encoding $i_{3}\cdots i_0$ are,
respectively, listed in the first and second column of~\Cref{posProofExample}.
 For any tuple $b = (b_l,\ldots,b_u)$ of $N$-qubit observables and
any tuple $i = (i_{l},\ldots,i_u)$ of binary digits with the same length as $b$,
let $b^i$ denote the matrix product $b_l^{i_l}.\cdots.b_u^{i_u}$ in this order.
(Despite the multiplicative notation, it is a linear combination of the elements
of $b$, with $0$ or $1$ coefficients.)

Let $b = (b_k,\ldots,b_0)$ be a basis of $S$. Its $k+1$ elements are independent
and mutually commuting $N$-qubit observables, chosen without phase.
 In~\Cref{posProofExample}, the basis $b$ is $(ZIXX,XZXY,IXXX,IIYZ)$. 
The value of each $b^i$ is displayed in the third column. The basis vectors $b_0$,
$b_1$, $b_2$, and $b_3$, respectively, appear in the rows of this column indexed by
$i= 1, 2, 4$ and $8$. With the former conventions and notations, the
elements of $S$ are the norms of all the linear combinations $b^i$ for $0 \leq
i \leq 2^{k+1}-1$. So, formally, the positivity property to prove is
\begin{align}
\prod_{0 \leq i < 2^{k+1}} \left|b^i\right| & = I^{\otimes N}, \label{prodId}
\end{align}
where $\prod_{\ell \leq i \leq u} a_i$ denotes the generalized matrix product
$a_{\ell} . \cdots . a_u$ of all the elements in the finite sequence
$(a_i)_{\ell \leq i \leq u}$ of $N$-qubit observables, in this order.

For $1 \leq j \leq N$, let $a_j$ denote the $j$-th qubit of the phase-free
$N$-qubit observable $a \in M^{\otimes N}$. For $0 \leq m \leq k$, the $j$-th
qubit of the $m$-th basis vector $b_m$ is denoted $b_{m,j}$. As a direct
consequence, for all $s,t$ in ${\cal P}^{\otimes N}$,
\begin{align}
\label{jthNormEq}\left|s.t\right|_j = \left||s|_j.|t|_j\right|.
\end{align}

The matrix product of all the elements of $S$ (left-hand side
of~(\ref{prodId})) is the tensor product for all $j \in [1..N]$ of the
matrix products of their $j$-th qubits. Formally, 
\begin{align*}
\prod_{0 \leq i < 2^{k+1}} \left|b^i\right|
& = \bigotimes\limits_{1 \leq j \leq N} \, \prod_{0 \leq i < 2^{k+1}} \left|b^i\right|_j.
\end{align*}
So, we prove that it is $I^{\otimes N}$ as a consequence of the (stronger) fact
that all the products of their $j$-th qubits equal $I$ (Lemma~\ref{qubitPosLemma}),
when the order of the qubit products is chosen to be the \emph{lexicographic
order} on the tuples $(i_{m-1},\ldots,i_0)$, corresponding to the usual order
$<$ on the natural numbers they encode. Since $S$
comes from a totally isotropic space, all its pairs of elements mutually
commute, so the order of the product of all its elements (left-hand side
of~(\ref{prodId})) can be chosen to be this lexicographic order.

 In~\Cref{posProofExample}, the $j$-th qubit $\left|b^i\right|_j$
of the elements of $S$ are shown in Columns 4, 6, 8, and 10, for $j = 3,2,1,$ and
$0$. The matrix product of all elements in each of these columns is expected to
be $I$. The remaining contents of~\Cref{posProofExample} are described
in the proof of~Lemma~\ref{qubitPosLemma}, postponed in
Appendix~\ref{qubitwisePosSec}  to lighten the present proof, which
ends here.
\end{proof}

\begin{lemma}\label{qubitPosLemma}
For $3 \leq k \leq N-1$, let $S$ be a vector subspace of $M^{\otimes N}$, of
dimension $k+1$, generated by a basis $(b_k,\ldots,b_0)$ of $k+1$ independent
and mutually commuting vectors without phase. Let $L = (\left|b^i\right|)_{0 
\leq i < 2^{k+1}}$ be the sequence of all the elements of $S$ in lexicographic
order.
Then, for all  $1 \leq j \leq N$, the matrix product 
$\prod_{0 \leq i < 2^{k+1}} \left|b^i\right|_j$
\noindent of all the elements of the sequence of $j$-th qubits in $L$, in
the same order as in $L$, always equals $I$.
\end{lemma}

\subsection{Geometrical constraints and a related corollary on the sign(s) of 
  three-dimensional subspaces}

In addition to the algebraic proof of~Proposition~\ref{kgeq3pos-prop} presented
in~\Cref{tisGeq3Sec}, this section provides a geometric interpretation of the
positivity of every PG$(3,2)$ of a multi-qubit $W_N$, for $N \geq 4$.

\begin{proposition}
\label{PG32negProp}
Let $P$ be a PG$(3,2)$ of a multi-qubit $W_N$ with $N \geq 4$. Then,
the following four propositions are equivalent:
\begin{itemize}
\item[a)] $P$ is negative.
\item[b)] Each of the 56 spreads of lines of $P$ contains an odd number of
negative lines.
\item[c)] Through each of the 35 lines of $P$ there pass an odd number of
negative planes.
\item[d)] Through each of the 15 points of $P$ there pass an odd number of
negative lines.
\end{itemize}
\end{proposition}

\begin{proof}
As the observables of any subspace of $W_N$ mutually commute, we can
take any product of observables in $P$ in any order.

a) $\Leftrightarrow$ b) A spread of lines of $P$ is a set of five pairwise
disjoint lines that partition the set of points of $P$. Pick up a spread of
lines, take first the product of the three observables on each line of the
spread and then multiply the five products obtained to get the sign of $P$. This
space will be {\it negative} (i.\,e., the product of all the 15 observables will
be equal to $-I^{\otimes N}$) iff the number of negative lines in the spread
selected is {\it odd}. Clearly, this property must hold irrespectively of which
spread of lines we select.

a) $\Leftrightarrow$ c) Let us assume that the 15 points/observables of $P$ are
labeled simply as $\{a$, $b$, $c$, $d$, $e$, $f$, $g$, $h$, $i$, $j$, $k$, $l$,
$m$, $n$, $o\}$, that $a,b,c$ lie on a line and that the three planes through
this line are $F_1 = \{a,b,c; d,e,f,g\}$, $F_2 = \{a,b,c; h,i,j,k\}$, and $F_3 =
\{a,b,c; l,m,n,o\}$. Let us consider the following product of the
products of all the points in each of these three planes:
\begin{align}
\label{Theta}
(a.b.c.d.e.f.g).(a.b.c.h.i.j.k).(a.b.c.l.m.n.o).
\end{align}
By commutation relations and $a^2 = b^2 = c^2 = I^{\otimes N}$, this
product~(\ref{Theta}) equals the product
\begin{align}
\label{prodP}
a.b.c.d.e.f.g.h.i.j.k.l.m.n.o,
\end{align}
of all the points of $P$. So, $P$ is negative iff one or all the three planes 
through the line in question are negative.

Again, as the sign of the space is an invariant this reasoning must hold for the
planes through any line of $P$.

a) $\Leftrightarrow$ d) Let us assume that the seven lines passing via the
point/observable $a$ are $\{a,b,c\}$, $\{a,d,e\}$, $\{a,f,g\}$, $\{a,h,i\}$,
$\{a,j,k\}$, $\{a,l,m\}$, and $\{a,n,o\}$. Consider the following
product of the products of all the points in these lines:
\begin{align}
\label{ThetaPrime} (a.b.c).(a.d.e).(a.f.g).(a.h.i).(a.j.k).(a.l.m).(a.n.o).
\end{align}
After simplification, it again equals the product~(\ref{prodP}) of all the
points of $P$. So, $P$ is negative iff the number of negative lines passing
through $a$ is odd.

Again, this must hold irrespectively of the reference point/observable chosen.
\end{proof}

As we have proved (see~Proposition~\ref{kgeq3pos-prop}) that there are no negative
PG$(3,2)$'s in $W_N$, for $N \geq 4$, Proposition~\ref{PG32negProp} leads to the
following geometrically slanted corollary:

\begin{corollary}
\label{PG32posConjec}
As a PG$(3,2)$ of a multi-qubit $W_N$, $N \geq 4$ is always positive,
then: a) each of its 56 spreads of lines contains an even number of negative
lines; b) through each of its 35 lines there pass an even number of negative
planes; and c) through each of its 15 points there pass an even number of
negative lines.
\end{corollary}

\subsection{No analogs of Mermin pentagrams in the 4-qubit case}

Given the fact, stemming from~Proposition~\ref{kgeq3pos-prop}, that all generators -- that is,
PG$(3, 2)$s -- of the four-qubit $\mathcal{W}(7, 2)$ are positive, there will be
no analogs of (three-qubit) Mermin pentagrams in the four-qubit case. We
illustrate this fact in Figure~\ref{nm}. We took a particular quadratic
$\mathcal{W}(5, 2)$ having 42 negative planes. In this $\mathcal{W}(5, 2)$, we
picked up one spread of planes; this spread features four negative planes and is
illustrated on the left-hand side of the figure. Now, through each plane of
$\mathcal{W}(7, 2)$ there are three PG$(3, 2)$s, and the totality of 27 such
spaces can be split into 3 sets of 9 elements each. One set has a special
footing as all its PG(3, 2)s pass through a common point -- the pole of that PG$
(6, 2)$ where our $\mathcal{W}(5, 2)$ is located in (being also the nucleus of
the $\mathcal{W}(5, 2)$) -- and it will be disregarded. The other two
nine-element sets are similar in the sense that each of them consists of
mutually disjoint PG$(3, 2)$s. The affine parts, AG$(3, 2)$s, of these spaces of
either set, however, cover the same  ($9 \times 8 =$) 72 points; this means that
each of these 72 points will lie in two different AG$(3, 2)$s. In analogy with a
Mermin pentagram, these ($9 \times 2=$) 18 AG$(3, 2)$s will be our contexts.
They are shown on the right-hand side of Figure~\ref{nm}; here, the two contexts
originating from the same plane of the spread are arranged in the row and/or the
column bearing the number of the corresponding plane. From this construction, it
is clear that {\it irrespectively} of the number of negative planes our selected
spread features, there will always be an {\it even} number of negative contexts
in our analog of a Mermin pentagram; indeed, each PG$(3, 2)$ -- being always
positive -- that passes via  a negative plane will have the corresponding AG$
(3, 2)$ negative, but for each plane these AG$(3, 2)$s come in pairs. A similar
argument can be used to show that there are no higher-rank analogs of Mermin
pentagrams for any $N > 3$.

\begin{figure}[pth!]
\centerline{\includegraphics[width=7.0truecm,clip=]{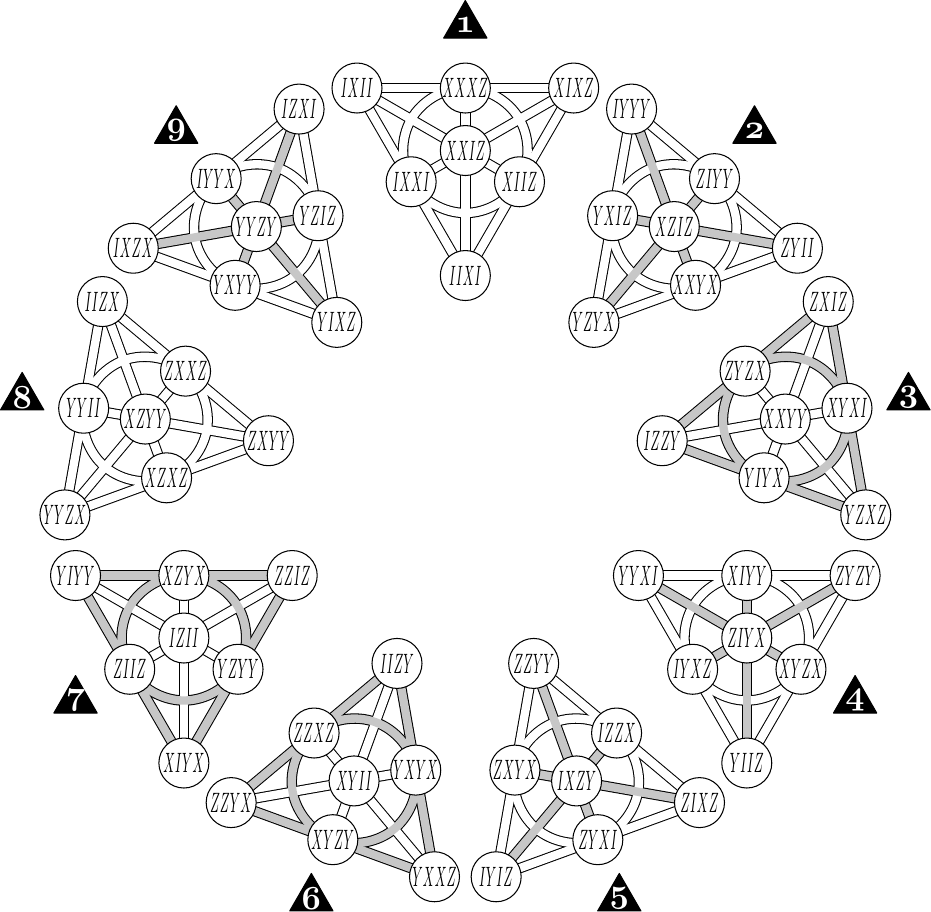}

\hspace{1.0cm}\includegraphics[width=6.0truecm,clip=]{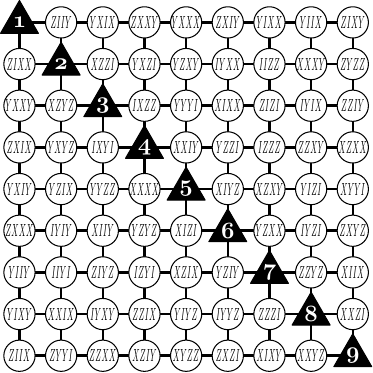}}

\vspace*{.2cm}
\caption{{\it Left:} A spread of planes in the selected four-qubit 
  $\mathcal{W}(5,2)$, with negative lines shown in gray; each plane featuring 
  three concurrent negative lines is negative. {\it Right:} The associated sets 
  of 18 contexts (9 rows and 9 columns), 2 per each plane of the spread.}
\label{nm}
\end{figure}

\subsection{Contextuality properties of two-spreads: The 
  contextuality degree of all \PDFstr{$N$}{N}-qubit two-spreads is 1}
\label{twospreadSec}
Apart from grids (aka Mermin--Peres magic squares), a multi-qubit doily hosts
another prominent, yet almost virtually unknown to the quantum informational
community, class of contextual configurations, the so-called two-spreads. Given
a multi-qubit doily, remove from it one spread of lines while keeping all the
points; what we get is a \emph{two-spread}, that is, a (sub-)configuration having
15 points and 10 lines, with 3 points per line and 2 lines per point(see,
\textit{e.g}., Figure 8 of \cite{polster}). Figure~\ref{5q-two-spread} is an 
illustration of a two-spread living in a quadratic five-qubit doily.

\begin{figure}[htb!]
\centerline{\includegraphics[width=6cm,angle=0,clip=true]{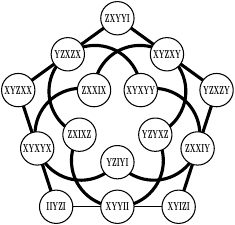}}
\vspace*{.2cm}
\caption{An example of five-qubit two-spread having nine (i.\,e., the maximum 
  possible number of) negative lines (boldfaced).}
\label{5q-two-spread}
\end{figure}

Let us briefly recall that a ``classical'' contextual point-line configuration
features (a) an even number of lines/contexts per point/observable and (b) an
odd number of negative lines/contexts~\cite[]{HS17}. We already see that a
multi-qubit two-spread meets the first condition (as there are two lines per
each of its points). Now we show that it also satisfies the second condition. To
this end in view, we readily infer from our computer-generated results (after
carefully inspecting Figure 2 of our recent paper \cite[]{MSGDH}) that if a doily
is endowed with an {\it odd} number of negative lines, each of its six spreads
of lines contains an {\it even} number of them, whereas if a doily features an
{\it even} number of negative lines, each of its six spreads shares an {\it odd}
number of them. That is, removal of a spread leads in both cases to a
configuration endowed with an {\it odd} number of negative lines, as envisaged.

The former reasoning is substantiated by our exhaustive computer-aided analysis
up to $N=5$.

In $W_N$, for $2 \leq N \leq 5$, we picked up a representative doily
from each type (defined by the signature of observables) and subtype (defined by
the pattern formed by negative lines) as classified in~\cite{MSGDH} and found
out that all six two-spreads in it were contextual. Moreover, we also found that
in each case the degree of contextuality was the same and equal to 1.
In fact, by analyzing the incidence matrix of the two-spread configuration, it
appears that all two-spreads are contextual and the vector $E$ of negative lines
is at Hamming distance $1$ of the image of the incidence matrix. In other words
we have the following proposition:

\begin{proposition}
A two-spread in an $N$-qubit doily is contextual for each $N \geq 2$, and its
degree of contextuality is always 1.
\end{proposition}

\begin{proof}
Consider the incidence matrix of the two-spread configuration:
$$A=\left(\begin{array}{ccccccccccccccc}
1 & 1 & 1 & 0 & 0 &0 &0 &0 &0 & 0 & 0 & 0 & 0 & 0 & 0\\
0 & 0 & 1 & 1 & 1 &0 &0 &0 &0 & 0 & 0 & 0 & 0 & 0 & 0\\
0 & 0 & 0 & 0 & 1 &1 &1 &0 &0 & 0 & 0 & 0 & 0 & 0 & 0\\
0 & 0 & 0 & 0 & 0 &0 &1 &1 &1 & 0 & 0 & 0 & 0 & 0 & 0\\
1 & 0 & 0 & 0 & 0 &0 &0 &0 &1 & 1 & 0 & 0 & 0 & 0 & 0\\
0 & 1 & 0 & 0 & 0 &0 &0 &0 &0 & 0 & 0 & 1 & 0 & 0 & 1\\
0 & 0 & 0 & 1 & 0 &0 &0 &0 &0 & 0 & 1 & 0 & 1 & 0 & 0\\
0 & 0 & 0 & 0 & 0 &1 &0 &0 &0 & 0 & 0 & 1 & 0 & 1 & 0\\
0 & 0 & 0 & 0 & 0 &0 &0 &1 &0 & 0 & 0 & 0 & 1 & 0 & 1\\
0 & 0 & 0 & 0 & 0 &0 &0 &0 &0 & 1 & 1 & 0 & 0 & 1 & 0\\
\end{array}\right).$$
The linear combinations of the rows of $A$ span the nine-dimensional space of
vectors in $\mathbb{F}_2^{10}$ having an even number of $1$s in their
decomposition in the standard basis. In fact, $\text{Im}(A)$ is the space of
length-$10$ binary vectors with an even number of $1$s. Thus, any vector $E$
representing a context evaluation with an odd number of negative lines, that is, an
odd number of $1$s, is at Hamming distance $1$ of $\text{Im}(A)$. This proves that
two-spreads are always contextual with contextuality degree $1$.
\end{proof}

\subsection{Contextuality degree of three-qubit lines}

A very instructive result in Table~\ref{resultsTable} concerns the contextuality
degree of lines of $W_3$. This result contradicts Lemma 3 and Figure 1
of~\cite{Cab10}. Indeed the contextuality degree corresponds to the minimal
number of constraints that cannot be satisfied by any NCHV. In~\cite{Cab10}, it
is supposed to be always equal to the number of negative lines when we consider,
for a fixed $N$, all $N$-qubit observables and all lines. But the model obtained
by the SAT solver was able to satisfy, in the three-qubit case, $315-63=252$
constraints instead of $235$ as previously obtained. Note that this
new result changes the contextuality bound experimentally tested
in~\cite{holweck2021testing}. The experimental evaluation
of~\Cref{noncontextIneq} on IBM quantum computers described
in~\cite{holweck2021testing} provided a result of $236$, which still violates
the classical bound, now known to be $b=315-2\times 63=189$, instead of the
erroneous value $315-2\times 90=135$.

\begin{figure}[b!]
\centerline{\includegraphics[width=11truecm,trim={0 5.5cm 0 3cm},clip=]{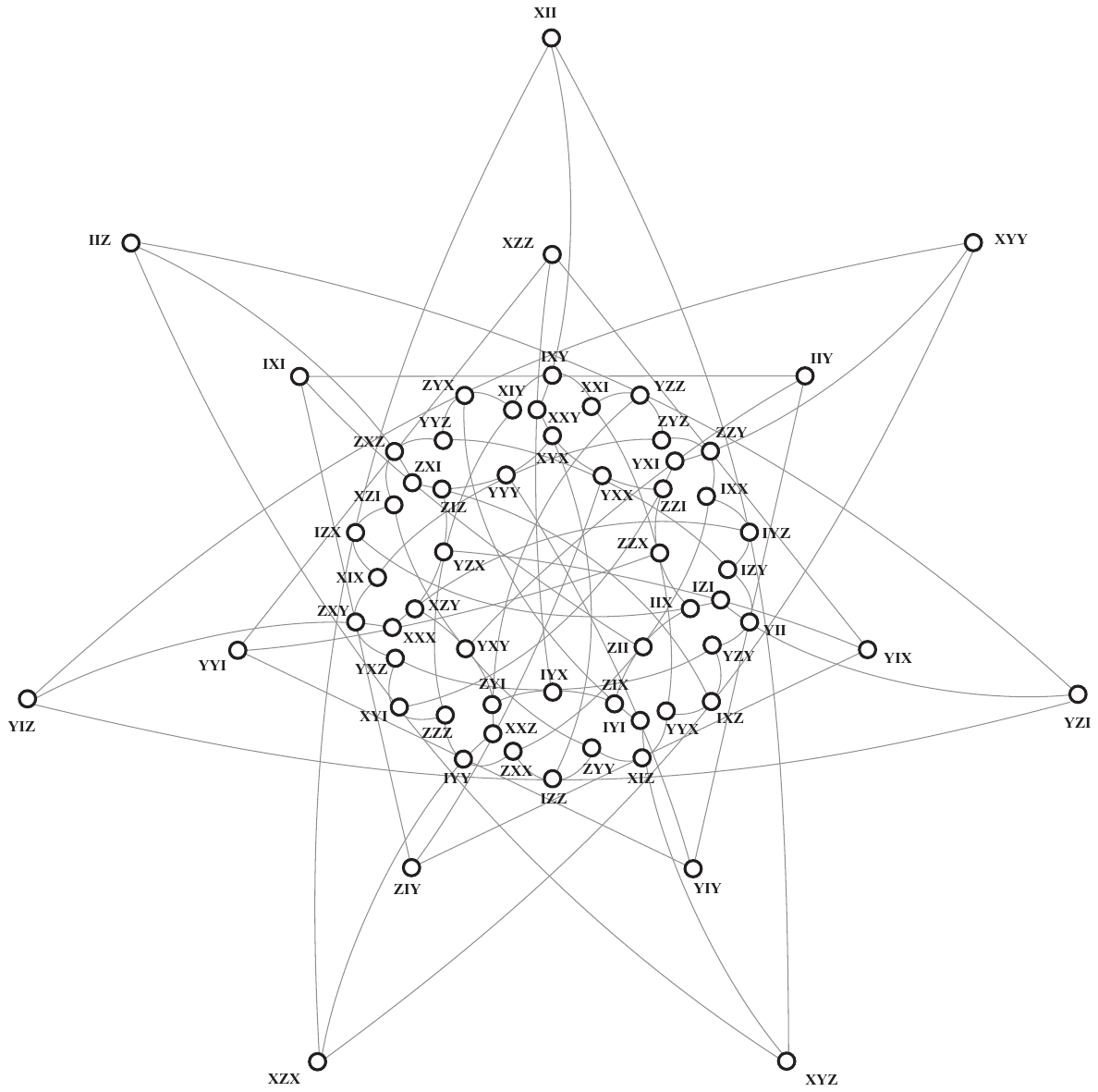}}
\caption{A classically embedded three-qubit split Cayley hexagon of order 2 
  that accommodates all the 63 invalid constraints for a particular solution
  found by the SAT solver. The graphical illustration of the hexagon is a
  simplified reproduction of that given in~\cite{psm}.}
\label{hexagon}
\end{figure}

With the C implementation of the contextuality degree algorithm~\ref{cdegalgo},
we managed to compute the contextuality degree of three-qubit lines. The
intermediate partial solutions found had the following Hamming distances: $89$,
$86$, $79$, $72$, and $63$. This means that it first checked for a Hamming distance
that is at most the number of negative lines, 90, and found a solution with 89
invalid constraints, and then it checked for at most one less invalid constraints
and found 86 of them. This process was repeated until it found a configuration
with 63 invalid constraints, after which the solver found no solution when asked
to let at most 62 invalid constraints satisfied. Thus, the contextuality
degree of the configuration of the lines of $W_3$ is 63.

It is worth concluding this section by mentioning that in this case we carried
out several runs with the SAT solver and, strikingly, in each case the set of 63
invalid constraints/lines was found to be isomorphic to a copy of the split
Cayley hexagon of order 2 when embedded classically into $W_3$; the
copy for one of these sets is illustrated in Figure~\ref{hexagon}.  This
finding sheds a rather unexpected new light on the contextuality role of the
hexagon in the three-qubit symplectic polar space discovered in~\cite{hdbs} and
will be treated in more detail in a separate paper.

\section{Conclusion}
\label{conclusionSec}

One of the most intriguing outcomes of our work is the fact that there are no
negative subspaces of dimension 3 and higher. This has a couple of
interesting implications. The first one concerns Mermin--Peres squares and Mermin
pentagrams. In the former case, the contexts correspond to generators of the
ambient $\mathcal{W}(3,2)$, while in the latter case the contexts represent
affine parts of generators of $\mathcal{W}(5,2)$. The non-existence of negative
generators for $k \geq 3$ then simply means that there will be no
higher-dimensional analogs of these two important classes of contextual
configurations in polar spaces of rank 5 and higher. The second implication
concerns contextuality properties of $s$-qubit $\mathcal{W}(2s-1,2)'$s living in
the ambient $N$-qubit $\mathcal{W}(2N-1,2)$, for $s < N$. We have already
verified that all multi-qubit doilies ($s=2$) are contextual for any $N$, and
that their degree of contextuality is 3~\cite[Proposition 1]{MSGDH}. As the next
natural step, it would be vital to address the corresponding behavior of rank
three and, especially, rank four polar spaces.

Next, given a multi-qubit $\mathcal{W}(2N-1,2)$, remove from it a spread of
generators while keeping all the observables; what we get is an analog of a
doily's two-spread. It  would be desirable to check, at least for a few randomly
selected cases, if these higher-dimensional analogs of two-spreads of a
multi-qubit doily are contextual, or not.

A final prospective task is to address the contextuality of quadrics when
instead of lines we consider their {\it planes} as contexts. A particular
worth-deserving object in this respect is a hyperbolic quadric of the four-qubit
symplectic polar space. This quadric, sometimes called the triality quadric,
possesses an unexpected threefold symmetry and, in addition, it hosts 960
ovoids, each of them being related via a particular Grassmannian mapping to a
spread of planes in the three-qubit polar space.

\paragraph*{Acknowledgments.}
\label{sec:acknowledgments}

We thank Zsolt Szab\'o for his help with figures~\ref{plig}, \ref{nm}, and~\ref{5q-two-spread}.

\paragraph*{Financial Support.}

This project is supported by the PEPR integrated project EPiQ ANR-22-PETQ-0007
part of Plan France 2030, by the project TACTICQ of the EIPHI Graduate School
(contract ANR-17-EURE-0002). This work also received a partial support from the
Slovak VEGA grant agency, Project 2/0004/20.

\paragraph*{Competing Interests.}

The authors declare no competing interests.

\printbibliography

\clearpage

\appendix

\section{Proof of~Lemma~\ref{qubitPosLemma}}
\label{qubitwisePosSec}

\begin{proof}
We still write $\prod_{0 \leq i < 2^m}$ for the product in the lexicographic
order, or sometimes use the more compact notation $\prod_{|i|=m}$, where $|i|$
denotes the length of the tuple $i$. Let us prove that the product in
lexicographic order of the $j$-th qubit of the norm of all linear combinations
of length $|i|=k+1$ of the basis vectors, i.e., of all elements of $S$, equals
the identity matrix. Formally,
\begin{align}
\prod_{|i|=k+1} |b^i|_j = & I \label{qubitProdId}
\end{align}
for all $1 \leq j \leq N$.

If $i$ is the $k$-tuple $(i_{k},\ldots,i_{1})$, let $2i$ denote the
$(k+1)$-tuple $(i_{k},\ldots,i_{1},0)$ and $2i+1$ denote the $(k+1)$-tuple
$(i_{k},\ldots,i_{1},1)$. For two tuples $i$ and $i'$ of bits (in
$\mathbb{F}_2$) with the same length, let $i \oplus i'$ denote the bitwise
addition (``exclusive or'') of their bits. From these notations, it comes that
$2i \oplus 2i' = 2(i \oplus i')$.

By associating each $\left|b^{2i}\right|_j$ with its lexicographic successor
$\left|b^{2i+1}\right|_j$ in the product, we get
\begin{align*}
\prod_{|i|=k+1} \left|b^i\right|_j
 = & \prod_{|i|=k} \left(\left|b^{2i}\right|_j.\left|b^{2i+1}\right|_j\right)
\end{align*}
\noindent for $k \geq 1$. Since $b^{2i+1} = b^{2i}.b_0$ and $b^{2i}.b^{2i} =
I^{\otimes N}$, all products
$\left(\left|b^{2i}\right|_j.\left|b^{2i+1}\right|_j\right)$ equal $b_{0,j}$
multiplied by some complex phase that will be precised below
(in~\Cref{posProofExample}, this product is shown in Columns 5, 7, 9 and 11).
This is a key property for the end of the proof, because the product of all
these $b_{0,j}$ with their phase commute, which allows us to re-order them in
any order, until showing that the product of all these phases is $1$.

A trivial case is when $b_{0,j} = I$. It is illustrated by Columns 4-7
of~\Cref{posProofExample}, for $j = 2$ and $j = 3$. In this case all the
products $\left(\left|b^{2i}\right|_j.\left|b^{2i+1}\right|_j\right)$ of two
consecutive qubits equal $I$, as detailed below:
\begin{align*}
\left|b^{2i}\right|_j.\left|b^{2i+1}\right|_j
& = \left|b^{2i}\right|_j.\left|b^{2i}.b_0\right|_j 
  = \left|b^{2i}\right|_j.\left|\left|b^{2i}\right|_j.b_{0,j}\right| 
  = \left|b^{2i}\right|_j.\left|\left|b^{2i}\right|_j.I\right| = \left(\left|b^{2i}\right|_j\right)^2 = I.
\end{align*}
\noindent Thus the complete product also equals $I$. Apart from this trivial
case, the following two complementary cases can be considered.

\paragraph{Case 1:} Consider first the case when the inequality
\begin{align}
\left|b^{2r}\right|_j \neq b_{0,j}
\label{case1koch}
\end{align}
holds for all tuples of coefficients $r$ of length $|r|=k$. This case is
illustrated by Columns 8-9 of~\Cref{posProofExample}, for $j=1$ and $b_{0,1} =
Y$.

Let $i = (i_k,\ldots,i_1)$ be any tuple of binary coefficients of length
$|i|=k$. The inequality (\ref{case1koch}) with $r=0$ implies that $b_{0,j}$ is
different from $I$. So, the set $\{X,Y,Z\}$ of Pauli matrices is composed of
$b_{0,j}$ and the other two Pauli matrices, hereafter denoted by $U$ and $V$.
For instance, in Columns 8-9 of~\Cref{posProofExample}, for $j=1$, $b_{0,1} =
b_{0,j} = Y$, so $U$ can be $X$ and $V$ can be $Z$.

Now, assume that there are two tuples $i$ and $i'$ of binary coefficients of
length $k$ such that $\left|b^{2i}\right|_j = U$ and $\left|b^{2i'}\right|_j =
V$. Then $\left|b^{2(i\oplus i')}\right|_j = \left|b^{2i\oplus 2i'}\right|_j =
\left|b^{2i}. b^{2i'}\right|_j = \left|\left|b^{2i}\right|_j.
\left|b^{2i'}\right|_j\right| = \left|U.V\right| = b_{0,j}$, which contradicts
the inequality (\ref{case1koch}). Consequently, there is a unique Pauli matrix,
different from $b_{0,j}$, hereafter assumed to be $U$ without loss of
generality, such that the $j$-th qubit of each vector $\left|b^{2i}\right|$ with
$|i|=k$ is either $I$ or $U$. For instance, in Columns 8-9
of~\Cref{posProofExample}, $U$ is $X$. It follows that any pair
$(\left|b^{2i}\right|_j, \left|b^{2i+1}\right|_j) = (\left|b^{2i}\right|_j,
\left|b^{2i}.b_0\right|_j)$ of consecutive $j$-th qubits whose first element
comes from a linear combination without $b_{0}$ is either $(I,b_{0,j})$ (for $i
= 000, 011, 101$ and $110$ in~\Cref{posProofExample}) or
$(U,\left|U.b_{0,j}\right|)$ (for $i = 001, 010, 100$ and $111$
in~\Cref{posProofExample}, gray cells in Column 9).

The sets of pairs in these two cases are in one-to-one correspondence according
to the involution $W \mapsto \left|U.W\right|$, so their cardinality is half the
total cardinality, i.e., $2^k/2=2^{k-1}$.

In the first case, $\left|b^{2i}\right|_j.\left|b^{2i+1}\right|_j = I.b_{0,j} =
b_{0,j}$, without phase. In the second case,
$\left|b^{2i}\right|_j.\left|b^{2i+1}\right|_{j} = U.\left|U.b_{0,j}\right| =
U.\left|\pm\text{i}\, V\right|$, because the product of two distinct Pauli
matrices always yields the third one, which is $V$ here, up to a phase which is
either \text{i} or $-\text{i}$. Finally,
$\left|b^{2i}\right|_j.\left|b^{2i+1}\right|_j = U.V = p\, b_{0,j}$ for the same
reason, with the same phase $p$ among \text{i} and $-\text{i}$ for all these
products. For instance, this phase $p$ is $-\text{i}$ in the gray cells in
Column 9 of~\Cref{posProofExample}.

By separating the complete product along these two cases, it comes
\begin{align*}
\prod_{|i|=k+1} \left|b^i\right|_j
 = & \prod_{|i|=k} \left|b^{2i}\right|_j.\left|b^{2i+1}\right|_j
   = \left(\prod_{|i|=k,\ \left|b^{2i}\right|_j=I} \left|b^{2i}\right|_j.\left|b^{2i+1}\right|_j\right)
     \left(\prod_{|i|=k,\ \left|b^{2i}\right|_j=U} \left|b^{2i}\right|_j.\left|b^{2i+1}\right|_j\right)
\\
 = & \left(\prod_{|i|=k,\ \left|b^{2i}\right|_j=I} b_{0,j}\right)
     \left(\prod_{|i|=k,\ \left|b^{2i}\right|_j=U}p\, b_{0,j}\right)
   = p^{2^{k-1}}\, \left(\prod_{|i|=k} b_{0,j}\right)
   = (\pm \text{i})^{2^{k-1}}\, b_{0,j}^{2^k} = 1 \, I = I 
\end{align*}
when $k \geq 3$.

\paragraph{Case 2:} Finally, consider the complementary case when there is some
tuple of coefficients $r \neq 0$ of length $|r|=k$ such that
$\left|b^{2r}\right|_j = b_{0,j}$. This case is illustrated by Columns 10-11
in~\Cref{posProofExample}, with $j=0$ and $r = 011$, since
 $\left|b^{0110}\right|_0 = b_{0,0} = Z$ in this example ($r = 110$ could be
 another valid choice).

In $r = (r_{k-1}, \ldots, r_0) \neq 0$, let $0 \leq q < k$ be any non-zero
coefficient ($r_q = 1$). For instance, since $r = (0,1,1)$
in~\Cref{posProofExample} for $j=0$, let us choose $q = 1$ ($q=0$ would also
hold) in this example. Let $R_0 = \{ i \in [0..2^k-1] \ | \ i_q = 0\}$ be the
subset of numbers in $[0..2^k-1]$ whose $q$-th bit is the opposite of the $q$-th
bit $r_q = 1$ in $r$. Then, $R_0$ and $([0..2^k-1] - R_0)$ partition
$[0..2^k-1]$ into two equal subsets, in bijection according to the involution $i
\mapsto r \oplus i$. For instance, in~\Cref{posProofExample}, $R_0$ is $\{ 000,
001, 100, 101\}$.

The product $\prod_{|i|=k} \left|b^{2i}\right|_j.\left|b^{2i+1}\right|_j$ is
re-organized according to $R_0$ and simplified as follows:

\begin{align*}
\prod_{|i|=k+1} \left|b^i\right|_j
 = &\prod_{|i|=k} \left|b^{2i}\right|_j.\left|b^{2i}.b_0\right|_j
 = \left(\prod_{|i|=k,\, i_q=0} \left|b^{2i}\right|_j.\left|b^{2i}.b_0\right|_j\right) \left(\prod_{|i|=k,\ i_q=1} \left|b^{2i}\right|_j.\left|b^{2i}.b_0\right|_j\right)
\\
 = &\left(\prod_{|i|=k,\, i_q=0} \left|b^{2i}\right|_j.\left|b^{2i}.b_0\right|_j\right) \left(\prod_{|i|=k,\ i_q=0} \left|b^{2(r\oplus i)}\right|_j.\left|b^{2(r\oplus i)}.b_0\right|_j\right)
\\
 = &\left(\prod_{|i|=k,\,i_q = 0} \left|b^{2i}\right|_j.\left|b^{2i}.b_0\right|_j .\left|b^{2r}.b^{2i}\right|_j.\left|b^{2r}.b^{2i}.b_0\right|_j\right).
\end{align*}

For instance, in~\Cref{posProofExample}, for $i = 001$, the subproduct
$\left|b^{2i}\right|_j.\left|b^{2i}.b_0\right|_j
.\left|b^{2r}.b^{2i}\right|_j.\left|b^{2r}.b^{2i}.b_0\right|_j$ in this product is
\begin{align*}
\left|b^{0010}\right|_0.\left|b^{0010}.b_0\right|_0
.\left|b^{0110}.b^{0010}\right|_0.\left|b^{0110}.b^{0010}.b_0\right|_0 & = 
X.\left|X.Z\right|.\left|Z.X\right|.\left|Z.X.Z\right|
\\
& = X.(Y.Y).X = X.(I).X = I.
\end{align*}

This reduction to $I$ holds in all generality, as detailed below. 
From~(\ref{jthNormEq}) and commutation relations,  it comes that
\begin{align*}
\left|b^{2r}.b^{2i}\right|_j
 & = \left|\left|b^{2r}\right|_j.\left|b^{2i}\right|_j\right|
   = \left|b_{0,j}.\left|b^{2i}\right|_j\right|
   = \left|\left|b^{2i}\right|_j.b_{0,j}\right|
   = \left|b^{2i}.b_0\right|_j
\end{align*}
so
\begin{align*}
\prod_{|i|=k+1} \left|b^i\right|_j
 & = \left(\prod_{|i|=k,\,i_q = 0} \left|b^{2i}\right|_j.\left(\left|b^{2i}.b_0\right|_j\right)^2.\left|b^{2r}.b^{2i}.b_0\right|_j\right)
   = \left(\prod_{|i|=k,\,i_q = 0} \left|b^{2i}\right|_j.\left|b^{2r}.b^{2i}.b_0\right|_j\right).
\end{align*}
Similarly,
\begin{align*}
\left|b^{2r}.b^{2i}.b_0\right|_j
 & = \left|\left|b^{2r}\right|_j.\left|b^{2i}\right|_j.b_{0,j}\right|
   = \left|b_{0,j}.\left|b^{2i}\right|_j.b_{0,j}\right|
   = \left|b_{0,j}^2.\left|b^{2i}\right|_j\right|
   = \left|b^{2i}\right|_j
\end{align*}
so
\begin{align*}
\prod_{|i|=k+1} \left|b^i\right|_j
 & = \prod_{|i|=k,\,i_q = 0} \left(\left|b^{2i}\right|_j\right)^2
   = \prod_{|i|=k,\,i_q = 0} I = I
\end{align*}
that completes the proof.
\end{proof}

\end{document}